\documentclass{article}%
\usepackage{amsmath}
\usepackage{amsfonts}
\usepackage{amssymb}
\usepackage{graphicx}
\usepackage{float}
\usepackage{algorithm}
\usepackage{placeins}%
\providecommand{\U}[1]{\protect\rule{.1in}{.1in}}
\numberwithin{equation}{section}
\newtheorem{theorem}{Theorem}[section]

\newtheorem{corollary}[theorem]{Corollary}

\newenvironment{proof}[1][Proof]{\noindent \textbf{#1.} }{\  \rule{0.5em}{0.5em}}
\begin{document}

\title{A Note on the Nash Equilibria of Some\\Multi-Player Reachability / Safety Games }
\author{Ath. Kehagias}
\date{\today}
\maketitle

\begin{abstract}
In this short note we study a class of multi-player, turn-based games with
deterministic state transitions and reachability / safety objectives (this
class contains as special cases \textquotedblleft classic\textquotedblright%
\ two-player reachability and safety games \ as well as multi-player and
\textquotedblleft\textquotedblleft stay--in-a-set\textquotedblright\ and
\textquotedblleft reach-a-set\textquotedblright\ games). Quantitative and
qualitative versions of the objectives are presented and for both cases we
prove the existence of a deterministic and memoryless Nash equilibrium; the
proof is short and simple, using only Fink's classic result about the
existence of Nash equilibria for \emph{multi-player discounted stochastic
games}.

\end{abstract}

\section{Introduction\label{sec01}}

The simplest $\omega$-regular games are, arguably, two-player turn-based
safety and reachability games \cite{Mazala2002}. \emph{Multiplayer} variants
of these are the \textquotedblleft\emph{Stay--in-a-set}\textquotedblright%
\ (SIAS) games \cite{Maitra2003,Secchi2002} and \textquotedblleft%
\emph{Reach-a-set}\textquotedblright\ (RAS) games
\cite{Chatterjee2003,Chatterjee2004a}. The existence of \emph{Nash equilibria}
(NE)\ has been proved: for SIAS games in \cite{Secchi2002} and for RAS games
in \cite{Chatterjee2003,Chatterjee2004a}; in particular in the special case of
turn-based games with deterministic state transitions and Borel objectives
(these include SIAS and RAS objectives) the existence of a pure \emph{Nash
equilibrium} (NE) is proved in \cite[Corollary 1]{Chatterjee2004a}. In both
cases the state space is assumed finite and the NE are not, in general,
\ memoryless. \footnote{In \cite[Theorem 1]{Chatterjee2004a} is also proved
the existence of memoryless $\varepsilon$-NE for a broader class, which
contains SIAS\ and RAS games.} A stronger result is proved in
\cite{Ummels2010}, namely: every turn-based multi-player game with
deterministic state transitions and Borel objectives possesses a pure
\emph{sub-game perfect} (and hence memoryless) equilibrium. These results are
quite general but their proofs are rather involved.

In the current note our main goal is to provide a \emph{short and simple}
proof of a special case: every turn-based SIAS and RAS\ game with
deterministic state transitions possesses a deterministic and memoryless NE.
This is proved using only Fink's classic result on the existence of NE for
\emph{multi-player discounted stochastic games} \cite{Fink1961}.

Our result is actually a little more general, in that it applies to the class
of multi-player, turn-based games with deterministic state transitions,
reachability objectives for some players and safety objectives for others. For
brevity, we will henceforth refer to these as \emph{multi-player reachability
/ safety games} (MPRS games); they contain as special cases classic
reachability and safety games as well as SIAS and RAS\ games.

Informally, the MPRS game can best be visualized as a \emph{graphical game},
in which $N$ players move a token along the arcs of a digraph $G=\left(
V,E\right)  $. The vertices of $G$ are partitioned into $N$ sets:\ $V=V_{1}%
\cup V_{2}\cup...\cup V_{N}$; if at the $t$-th turn the token is located on a
vertex $v_{t}\in V_{n}$, then it is moved by the $n$-th player (henceforth
denoted by $P_{n}$)\ into some vertex $v_{t+1}$ such that $\left(
v_{t},v_{t+1}\right)  $ is an arc of $G$. In general we have two type of
players: \emph{reachers} and \emph{avoiders}. To each $P_{n}$ is associated a
\emph{nonempty} set $R_{n}\subseteq V$, related to his objective. If $P_{n}$
is a reacher, he wins iff the token enters some vertex $v\in R_{n}$; if he is
an avoider, he wins iff the token never enters a vertex $v\in R_{n}$.

In Section \ref{sec02} we define the \emph{quantitative} MPRS game and prove
that every such game has a NE\ in deterministic memoryless strategies. In
Section \ref{sec03} we do the same things for the \emph{qualitative} MPRS game.

\section{The Quantitative MPRS Game\label{sec02}}

We now formulate MPRS as a \emph{discounted stochastic game}.\footnote{We
follow the formulation of \cite{Filar1997}, expanded to the multi-player
case.} In what follows the quantities $N$, $V$, $E$, $V_{1}$, ..., $V_{N}$,
$R_{1}$, ..., $R_{N}$ are the ones presented in the previous section.

\begin{enumerate}
\item The \emph{player set} is $\left\{  P_{1},P_{2},...,P_{N}\right\}  $ or,
for simplicity, $\left\{  1,2,...,N\right\}  $.

\item The \emph{state set} is $S:=V\cup\left\{  \overline{s}\right\}  $, where
$V$ is the vertex set of the previously mentioned $G=\left(  V,E\right)  $ and
$\overline{s}$\ is the \emph{terminal state.}

\item We define $\left\{  S_{1},...,S_{N}\right\}  $, a partition of $S$, as
follows: $S_{1}:=V_{1}\cup\left\{  \overline{s}\right\}  $, $S_{2}:=V_{2}$,
..., $S_{N}:=V_{N}$.

\item For $n\in\left\{  1,2,...,N\right\}  $, $P_{n}$'s \emph{target set} is
$R_{n}$; the \emph{total target set} is $R:=\cup_{m=1}^{N}R_{m}$.

\item $A_{n}\left(  s\right)  $\ denotes $P_{n}$'s \emph{action set} when the
game is at state $s$ and is defined by ($\lambda$ is the \textquotedblleft
trivial\textquotedblright\ move):%
\begin{align*}
\text{when }s &  \in S_{n}\backslash R:A_{n}\left(  s\right)  :=\left\{
s^{\prime}:\left(  s,s^{\prime}\right)  \in E\right\}  ;\\
\text{when }s &  \in S_{m}\backslash R,m\neq n:A_{n}\left(  s\right)
:=\left\{  \lambda\right\}  ;\\
\text{when }s &  \in R\cup\left\{  \overline{s}\right\}  :A_{n}\left(
s\right)  :=\left\{  \lambda\right\}  .
\end{align*}
$P_{n}$'s \textquotedblleft total\textquotedblright\ action set is
$A_{n}:=\cup_{s\in S}A_{n}\left(  s\right)  $.

\item The \emph{law of motion} is deterministic and has the following form%
\begin{align}
\text{when }s  &  \in S_{n}\backslash R\text{ and }a=\left(  \lambda
,...,a^{n},...,\lambda\right)  :\Pr\left(  s_{t+1}=s^{\prime}|s_{t}%
=s,a_{t}=a\right)  :=\left\{
\begin{array}
[c]{ll}%
1 & \text{when }s^{\prime}=a^{n}\text{, }\\
0 & \text{else;}%
\end{array}
\right. \label{eq001}\\
\text{when }s  &  \in R\cup\left\{  \overline{s}\right\}  \text{ and
}a=\left(  \lambda,...,\lambda,...,\lambda\right)  :\Pr\left(  s_{t+1}%
=s^{\prime}|s_{t}=s,a_{t}=a\right)  :=\left\{
\begin{array}
[c]{ll}%
1 & \text{when }s^{\prime}=\overline{s}\text{, }\\
0 & \text{else.}%
\end{array}
\right.  \label{eq001b}%
\end{align}
All admissible state/action combinations are covered by (\ref{eq001}%
)-(\ref{eq001b}), from which we see the following.

\begin{enumerate}
\item If the current state $s$ \textquotedblleft belongs\textquotedblright\ to
$P_{n}$ (i.e., $s\in S_{n}$) and is not a target state, then he is the only
player who can perform a non-trivial action $a^{n}\in V$; the next state is,
with certainty, $a^{n}$.

\item If the current state $s$ is either target or terminal, then the only
admissible action vector is $a=\left(  \lambda,...,\lambda,...,\lambda\right)
$; the next and all subsequent states are the terminal $\overline{s}$.
\end{enumerate}

It is convenient to describe the \emph{deterministic} state transitions in
terms of a state transition function $\mathbf{T}:S\times A\rightarrow S$,
defined by
\begin{equation}
\mathbf{T}\left(  s,a^{n}\right)  :=\left\{
\begin{array}
[c]{ll}%
a^{n} & \text{when }s\in S_{n}\backslash R\text{ and }a^{n}\in A^{n}\left(
s\right)  \backslash\lambda,\\
\overline{s} & \text{when }s\in R\cup\left\{  \overline{s}\right\}  \text{ and
}a^{n}=\lambda.
\end{array}
\right.  \quad\label{eq002}%
\end{equation}
All admissible state/action combinations are covered by (\ref{eq002}).

\item $P_{n}$'s \emph{turn payoff function} depends only on the current game
state $s$ (but not on the current action vector) and can be either of the
following:
\[
q^{n}\left(  s\right)  :=\left\{
\begin{array}
[c]{ll}%
1 & \text{when }s\in R_{n}\\
0 & \text{when }s\notin R_{n}%
\end{array}
\right.  \ \ \text{(}P_{n}\text{ is a reacher)};\qquad q^{n}\left(  s\right)
:=\left\{
\begin{array}
[c]{rl}%
-1 & \text{when }s\in R_{n}\\
0 & \text{when }s\notin R_{n}%
\end{array}
\right.  \ \ \text{(}P_{n}\text{ is an avoider).}%
\]
$P_{n}$'s \emph{total payoff} function is (with the \emph{discount factor}
$\gamma\in\left(  0,1\right)  $): $Q^{n}\left(  s_{0},s_{1},...\right)
=\sum_{t=0}^{\infty}\gamma^{t}q^{n}\left(  s_{t}\right)  $.
\end{enumerate}

The game starts at an initial state $s_{0}=s\in S\backslash\overline{s}$ and,
at the $t$-th turn ($t\in\left\{  0,1,2,...\right\}  $)\ all players perform
\textquotedblleft trivial\textquotedblright\ moves, except for the player who
\textquotedblleft owns\textquotedblright\ $s_{t}$. Two possibilities exist.

\begin{enumerate}
\item If a target state is entered at some time $t^{\prime}$ ($s_{t^{\prime}%
}=s^{\prime}\in R=\cup_{m=1}^{N}R_{m}$) the next and all subsequent states are
the terminal ($\forall t>t^{\prime}:s_{t}=\overline{s}$).\footnote{Hence,
while the game lasts an infinite number of turns, it \emph{effectively} ends
at $t^{\prime}.$} \ For each $n\in\left\{  1,...,N\right\}  $, $P_{n}$
receives total $_{{}}$payoff:
\[
Q^{n}\left(  s_{0},s_{1},...\right)  =\left\{
\begin{array}
[c]{rl}%
\gamma^{t^{\prime}} & \text{if }s^{\prime}\in R_{n}\text{ and he is a
reacher;}\\
-\gamma^{t^{\prime}} & \text{if }s^{\prime}\in R_{n}\text{ and he is an
avoider;}\\
0\,\,\,\, & \text{if }s^{\prime}\notin R_{n}\text{.}%
\end{array}
\right.
\]

\item If a target state is never entered ($\forall t:s_{t}\notin R$), the game
continues ad infinitum and all players receive zero payoff.
\end{enumerate}

\noindent A reacher (resp. avoider)$\ P_{n}$ wants the game to enter $R_{n}$
in the shortest (resp. longest)\ possible \ time. Hence the above defined
discounted stochastic game will be called \textquotedblleft\emph{quantitative
MPRS game}\textquotedblright.

A \emph{finite-length history} is a finite sequence of states (we omit player
actions, since they will not be needed in our proof\footnote{Besides they are
directly inferred from the states, due to the deterministic law of motion.}):
\[
h=s_{0}s_{1}...s_{k}\in\underset{k\text{ times}}{\underbrace{S\times
S\times...\times S}}\qquad\text{for some }k\in\left\{  1,2,...\right\}  ;
\]
the set of all finite-length histories is denoted by $H^{\ast}$. A
\emph{deterministic strategy} for the $n$-th player is a function $\sigma^{n}$
which assigns an action to each finite-length history: $\sigma^{n}:H^{\ast
}\rightarrow A_{n}$. A strategy $\sigma^{n}$ is called \emph{memoryless} if it
only depends on the current state, in which case we write (with a slight
notation abuse) $\sigma^{n}\left(  s_{0}s_{1}...s_{k}\right)  =\sigma
^{n}\left(  s_{k}\right)  $. A \emph{strategy profile} is a tuple
$\sigma=\left(  \sigma^{1},\sigma^{2},...,\sigma^{N}\right)  $ which specifies
one strategy for each player. As usual, $\sigma^{-n}=\left(  \sigma
^{j}\right)  _{j\in\left\{  1,2,...,N\right\}  \backslash\left\{  n\right\}
}$, so we can write $\sigma=\left(  \sigma^{n},\sigma^{-n}\right)  $. Since an
initial state $s_{0}$ and a deterministic strategy profile $\sigma$ determine
fully the history $s_{0}s_{1}s_{2}...$, the payoff function $Q^{n}\left(
s_{0},s_{1},...\right)  \ $will also be written as $Q^{n}\left(  s_{0}%
,\sigma\right)  $, $Q^{n}\left(  s_{0},\sigma^{1},...,\sigma^{N}\right)  $ or
$Q^{n}\left(  s_{0},\sigma^{n},\sigma^{-n}\right)  $.

\begin{theorem}
\label{prop0301}Every quantitative MPRS game has a deterministic memoryless
NE. In other words, there exists a profile of deterministic memoryless
strategies $\widehat{\sigma}=\left(  \widehat{\sigma}^{1},\widehat{\sigma}%
^{2},...,\widehat{\sigma}^{N}\right)  $ such that%
\begin{equation}
\forall n\in\left\{  1,2,...,N\right\}  ,\forall s_{0}\in S,\forall\sigma
^{n}:Q^{n}\left(  s_{0},\widehat{\sigma}^{n},\widehat{\sigma}^{-n}\right)
\geq Q^{n}\left(  s_{0},\sigma^{n},\widehat{\sigma}^{-n}\right)  .
\label{eq02011}%
\end{equation}
For every $s$ and $n$, let $u^{n}\left(  s\right)  :=Q^{n}\left(
s,\widehat{\sigma}\right)  $. Then the \ following equations are satisfied%
\begin{align}
\forall n,\forall s  &  \in S_{n}:\widehat{\sigma}^{n}\left(  s\right)
=\arg\max_{a^{n}\in A^{n}\left(  s\right)  }\left[  q^{n}\left(  s\right)
+\gamma u^{n}\left(  \mathbf{T}\left(  s,a^{n}\right)  \right)  \right]
,\label{eq02012}\\
\forall n,m,\forall s  &  \in S_{n}:u^{m}\left(  s\right)  =q^{m}\left(
s\right)  +\gamma u^{m}\left(  \mathbf{T}\left(  s,\widehat{\sigma}^{n}\left(
s\right)  \right)  \right)  . \label{eq02013}%
\end{align}

\end{theorem}

\begin{proof}
Fink has proved in \cite{Fink1961} that every $N$-player discounted stochastic
game has a memoryless NE in \emph{probabilistic} strategies; this result holds
for the general game (i.e., with \emph{concurrent} moves and probabilistic
strategies and state transitions). According to \cite{Fink1961}, at
equilibrium the following equations must be satisfied for all $m$ and $s$:
\begin{equation}
\mathfrak{u}^{m}\left(  s\right)  =\max_{\mathbf{p}^{m}\left(  s\right)  }%
\sum_{a^{1}\in A^{1}\left(  s\right)  }...\sum_{a^{N}\in A^{N}\left(
s\right)  }p^{1}\left(  a^{1}|s\right)  ...p^{N}\left(  a^{N}|s\right)
\left[  q^{m}\left(  s\right)  +\gamma\sum_{s^{\prime}}\Pi\left(  s^{\prime
}|s,a^{1},...,a^{N}\right)  \mathfrak{u}^{m}\left(  s^{\prime}\right)
\right]  , \label{eq02015}%
\end{equation}
where we have modified Fink's original notation to fit our own; in particular:

\begin{enumerate}
\item $\mathfrak{u}^{m}\left(  s\right)  $ is the expected value of
$u^{m}\left(  s\right)  $;

\item $p^{m}\left(  a^{m}|s\right)  $ is the probability that, given the
current game state is $s$, the $m$-th player plays action $a^{m}$;

\item $\mathbf{p}^{m}\left(  s\right)  =\left(  p^{m}\left(  a^{m}|s\right)
\right)  _{a^{m}\in A^{m}\left(  s\right)  }$ is the vector of all such
probabilities (one probability per available action);

\item $\Pi\left(  s^{\prime}|s,a^{1},a^{2},...,a^{N}\right)  $ is the
probability that, given the current state is $s$ and the player actions are
$a^{1},a^{2},...,a^{N}$, the next state is $s^{\prime}$ .
\end{enumerate}

\noindent Now choose any $n$ and any $s\in S_{n}$. For all $m\neq n$, the
$m$-th player has a single move: $A^{m}\left(  s\right)  =\left\{
\lambda\right\}  $, and so $p^{m}\left(  a^{m}|s\right)  =1$. Also, since
transitions are deterministic,%
\[
\sum_{s^{\prime}}\Pi\left(  s^{\prime}|s,a^{1},a^{2},...,a^{N}\right)
\mathfrak{u}^{n}\left(  s^{\prime}\right)  =\mathfrak{u}^{n}\left(
\mathbf{T}\left(  s,a^{n}\right)  \right)  .
\]
Hence, for $m=n$, (\ref{eq02015})\ becomes
\begin{equation}
\mathfrak{u}^{n}\left(  s\right)  =\max_{\mathbf{p}^{n}\left(  s\right)  }%
\sum_{a^{n}\in A^{n}\left(  s\right)  }p^{n}\left(  a^{n}|s\right)  \left[
q^{n}\left(  s\right)  +\gamma\mathfrak{u}^{n}\left(  \mathbf{T}\left(
s,a^{n}\right)  \right)  \right]  .\label{eq02016}%
\end{equation}
Furthermore let us define $\widehat{\sigma}^{n}\left(  s\right)  $ (for the
specific $s$ and $n$) by
\begin{equation}
\widehat{\sigma}^{n}\left(  s\right)  =\arg\max_{a^{n}\in A^{n}\left(
s\right)  }\left[  q^{n}\left(  s\right)  +\gamma\mathfrak{u}^{n}\left(
\mathbf{T}\left(  s,a^{n}\right)  \right)  \right]  .\label{eq02017}%
\end{equation}
If (\ref{eq02016}) is satisfied by more than one $a^{n}$, we set
$\widehat{\sigma}^{n}\left(  s\right)  $ to one of these arbitrarily. Then, to
maximize the sum in (\ref{eq02016}) the $n$-th player must set $p^{n}\left(
\widehat{\sigma}^{n}\left(  s\right)  |s\right)  =1$ and $p^{n}\left(
a^{n}|s\right)  =0$ for all $a^{n}\neq\widehat{\sigma}^{n}\left(  s\right)  $.
Since this is true for all states and all players (i.e., every player can,
without loss, use deterministic strategies) we also have $\mathfrak{u}%
^{n}\left(  s\right)  =u^{n}\left(  s\right)  $. Hence (\ref{eq02016})
becomes
\begin{equation}
u^{n}\left(  s\right)  =\max_{a^{n}\in A^{n}\left(  s\right)  }\left[
q^{n}\left(  s\right)  +\gamma u^{n}\left(  \mathbf{T}\left(  s,a^{n}\right)
\right)  \right]  =q^{n}\left(  s\right)  +\gamma u^{n}\left(  \mathbf{T}%
\left(  s,\widehat{\sigma}^{n}\left(  s\right)  \right)  \right)
.\label{eq02018}%
\end{equation}
For $m\neq n$, the $m$-th player has no choice of action and (\ref{eq02016})
becomes
\begin{equation}
u^{m}\left(  s\right)  =q^{m}\left(  s\right)  +\gamma u^{m}\left(
\mathbf{T}\left(  s,\widehat{\sigma}^{n}\left(  s\right)  \right)  \right)
.\label{eq02019}%
\end{equation}
We recognize that (\ref{eq02017})-(\ref{eq02019}) are (\ref{eq02012}%
)-(\ref{eq02013}); replacing $\mathfrak{u}^{n}\left(  \mathbf{T}\left(
s,a^{n}\right)  \right)  $ with $u^{n}\left(  \mathbf{T}\left(  s,a^{n}%
\right)  \right)  $ in (\ref{eq02017}) defines $\widehat{\sigma}^{n}\left(
s\right)  $ for every $n$ and $s$ and so yields the required deterministic
memoryless strategies $\widehat{\sigma}=\left(  \widehat{\sigma}^{1}%
,\widehat{\sigma}^{2},...,\widehat{\sigma}^{3}\right)  $.
\end{proof}

\section{The Qualitative MPRS Game\label{sec03}}

The qualitative MPRS game elements are identical to those of the quantitative
one, except for the payoff functions. The qualitative game: (i)\ does
\emph{not} have a turn payoff function; (ii)\ has total payoff function
\[
\widetilde{Q}^{n}\left(  s_{0},\sigma\right)  =\left\{
\begin{array}
[c]{rll}%
1 & \text{if} & Q^{n}\left(  s_{0},\sigma\right)  >0,\\
-1 & \text{if} & Q^{n}\left(  s_{0},\sigma\right)  <0,\\
0 & \text{if} & Q^{n}\left(  s_{0},\sigma\right)  =0;
\end{array}
\right.  \text{ }%
\]
It is easily checked that: $\widetilde{Q}^{n}\left(  s_{0},\sigma\right)  =1$
(resp. $\widetilde{Q}^{n}\left(  s_{0},\sigma\right)  =0$) iff $P_{n}$ is a
reacher (resp. an avoider) and his target set is entered (resp. not entered).
Accordingly, in the qualitative MPRS game $P_{n}$ \emph{wins} (resp.
\emph{loses})\ iff he achieves the maximum (resp. minimum)\ possible value of
$\widetilde{Q}^{n}$. More specifically, we have the following.

\begin{enumerate}
\item When $P_{n}$ is a reacher, he wins (resp. loses)\ iff $\widetilde{Q}%
^{n}\left(  s_{0},\sigma\right)  =1$ (resp. $\widetilde{Q}^{n}\left(
s_{0},\sigma\right)  =0$).

\item When $P_{n}$ is an avoider, he wins (resp. loses)\ iff $\widetilde
{Q}^{n}\left(  s_{0},\sigma\right)  =0$ (resp. $\widetilde{Q}^{n}\left(
s_{0},\sigma\right)  =-1$).
\end{enumerate}

In short, the quantitative $Q^{n}$'s defines the qualitative $\widetilde
{Q}^{n}$'s which are used to formalize win/lose criteria analogous to these of
reachability, safety, RAS and SIAS games; these are special cases of
qualitative MPRS:

\begin{enumerate}
\item two-player reachability games ($N=2$, $P_{1}$ a reacher with $R_{1}%
\neq\emptyset$ and $P_{2}$ an avoider with $R_{2}=R_{1}$);

\item safety games (same as the reachability game, with player roles interhanged);

\item SIAS\ games ($\forall n:P_{n}$ is an avoider);

\item RAS\ games ($\forall n:P_{n}$ is a reacher). \footnote{Note that, while
the RAS game can be seen as a \emph{variant} of the classic reachability game,
it is not a generalization thereof, because it does not involve a player with
safety objectives \cite{Chatterjee2003,Chatterjee2004a}. Similarly, the
classic safety game is not a SIAS\ game.}
\end{enumerate}

\noindent\noindent The most general MPRS game involves $N_{1}$ reachers and
$N_{2}$ avoiders; we can have more than one winners (e.g., if $P_{m}$ and
$P_{n}$ are reachers, both win if the token enters some $v\in R_{m}\cap
R_{n}\neq\emptyset$) and the same is true for losers.

It is easily checked that every $\widehat{\sigma}=\left(  \widehat{\sigma}%
^{1},...,\widehat{\sigma}^{N}\right)  $ which is a NE\ of the $Q^{n}$'s is
also a NE\ of the $\widetilde{Q}^{n}$'s. Hence, by Theorem \ref{prop0301}, we
have the following.

\begin{corollary}
Every qualitative MPRS game has a deterministic memoryless NE.
\end{corollary}


\begin{thebibliography}{9}                                                                                                %


\bibitem {Chatterjee2003}K. Chatterjee, R. Majumdar, and M. Jurdzi\'{n}ski.
\textquotedblleft On Nash equilibria in stochastic games.\textquotedblright%
\ \emph{Report No.UCB/CSD-3-1281}, Oct. 2003, Computer Science Division
(EECS), Univ. of California at Berkeley.

\bibitem {Chatterjee2004a}K. Chatterjee, R. Majumdar, and M. Jurdzi\'{n}ski.
\textquotedblleft On Nash equilibria in stochastic games.\textquotedblright%
\ \emph{International Workshop on Computer Science Logic}. Springer, Berlin,
Heidelberg, 2004.

\bibitem {Chatterjee2012}K. Chatterjee and T.A. Henzinger. \textquotedblleft A
survey of stochastic ù-regular games\textquotedblright. \emph{Journal of
Computer and System Sciences} 78.2 (2012): 394-413.

\bibitem {Filar1997}J. Filar and K. Vrieze. \emph{Competitive Markov Decision
Processes-Theory, Algorithms, and Applications}. Springer, 1997.

\bibitem {Fink1961}A.M. Fink. \textquotedblleft Equilibrium in a stochastic
$n$-person game\textquotedblright. \emph{Journal of science of the Hiroshima
university, series ai (mathematics)}, 28(1), pp.89-93, 1964.

\bibitem {Maitra2003}A. Maitra and W.D. Sudderth. \textquotedblleft Borel
stay-in-a-set games\textquotedblright. \emph{International journal of Game
Theory} 32.1 (2003): 97-108.

\bibitem {Mazala2002}R. Mazala. "Infinite games." \emph{Automata logics, and
infinite games}. pp. 23-38. Springer, 2002.

\bibitem {Secchi2002}P. Secchi and W. D. Sudderth. \textquotedblleft
Stay-in-a-set games\textquotedblright. \emph{International Journal of Game
Theory} 30.4 (2002): 479-490.

\bibitem {Ummels2010}M. Ummels. \emph{Stochastic multiplayer games: theory and
algorithms}. Amsterdam University Press, 2010.
\end{thebibliography}
\end{document}